\documentclass[english]{llncs}
\usepackage[T1]{fontenc}
\usepackage[latin9]{inputenc}
\usepackage{xcolor}
\usepackage{pdfcolmk}
\usepackage{array}
\usepackage{verbatim}
\usepackage{latexsym}
\usepackage{float}
\usepackage{mathtools}
\usepackage{amsmath}
\usepackage{amssymb}
\usepackage{stmaryrd}
\usepackage{setspace}
\PassOptionsToPackage{normalem}{ulem}
\usepackage{ulem}
\makeatletter

\providecommand{\tabularnewline}{\\}
\providecolor{lyxadded}{rgb}{0,0,1}
\providecolor{lyxdeleted}{rgb}{1,0,0}

\DeclareRobustCommand{\lyxsout}[1]{\ifx\\#1\else\sout{#1}\fi}

\usepackage{xcolor}
\definecolor{darkred}{rgb}{0.8,0,0} 
\usepackage{hyperref}
\hypersetup{final, colorlinks=true, linkcolor=darkred, urlcolor=darkred, citecolor=black, linkbordercolor=darkred, pdfborderstyle={/S/U/W 1}}
\usepackage{setspace}
\usepackage[absolute,overlay]{textpos} 
\usepackage{listings}
\usepackage{cite}
\usepackage{tikz}
\usetikzlibrary{arrows,positioning,fit,automata}

\tikzset{ 
world/.style={rectangle, draw=black, rounded corners, text width=50pt, minimum height=28pt, text centered},
real/.style={rectangle, draw=black, ultra thick, rounded corners, text width=50pt, minimum height=27pt, text centered},
agent/.style={rectangle, rounded corners, text width=6pt, text height=-1pt, text centered},
agentpred/.style={rectangle, rounded corners, text width=6pt, text height=-1pt, text centered, fill=black!20},
event/.style={rectangle, draw=black, text width=50pt, minimum height=22pt, text centered},
realevent/.style={rectangle, draw=black, ultra thick, text width=50pt, minimum height=22pt, text centered},
every loop/.style={max distance=10mm,looseness=10},
}

\makeatletter
\g@addto@macro \normalsize {%
 \setlength\abovedisplayskip{4pt plus 2pt minus 2pt}%
 \setlength\belowdisplayskip{4pt plus 2pt minus 2pt}%
}
\makeatother

\DeclareRobustCommand{\VAN}[2]{#2}

\AtBeginDocument{
  
}

\makeatother

\usepackage{babel}
\begin{document}
\setstretch{1}
\title{Dynamic Term-Modal Logic for Epistemic Social Network Dynamics (Extended Version)}
\author{Andr\'{e}s Occhipinti Liberman\inst{1} and Rasmus K. Rendsvig\inst{2}
\institute{DTU Compute   \\ \email{aocc@dtu.dk} \and Center for Information and Bubble Studies, University of Copenhagen \\ \email{rasmus@hum.ku.dk}}}
\maketitle
\begin{abstract}
Logics for social networks have been studied in recent literature. This paper presents a framework based on \emph{dynamic term-modal logic} ($\mathsf{DTML}$), a quantified variant of dynamic epistemic logic (DEL). In contrast with DEL where it is commonly known to whom agent names refer, $\mathsf{DTML}$ can represent dynamics with uncertainty about agent identity. We exemplify dynamics where such uncertainty and \emph{de re}/\emph{de dicto} distinctions are key to social network epistemics. Technically, we show that $\mathsf{DTML}$ semantics can represent a popular class of hybrid logic epistemic social network models. We also show that $\mathsf{DTML}$ can encode previously discussed dynamics for which finding a complete logic was left open. As complete reduction axioms systems exist for $\mathsf{DTML}$, this yields a complete system for the dynamics in question.
\end{abstract}

\section{Introduction}

Over recent years, several papers have been dedicated to logical studies
of social networks, their epistemics and dynamics \cite{Baltag_etal_2018,ChristoffHansen2013,ChristoffHansen2015,ChristoffHansenProietti2016,ChristoffRendsvig2014,ChristoffNaumov2018,Liang2011,Liu2014,Rendsvig2014,Seligman2011,Seligman2013,SmetsVelazquez2017,SmetsVelazquez2018a}.
The purpose of this literature typically is to define and investigate
some social dynamics with respect to e.g. long-term stabilization
or other properties, or to introduce formal logics that capture some
social dynamics, or both.

This paper illustrates how \emph{dynamic term-modal logic} ($\mathsf{DTML}$,
\cite{ALR}) may be used for the second purpose. In general, \emph{term-modal
logics} are first-order modal logics where the index of modal operators
are first-order terms. I.e., the operators double as predicates to
the effect that e.g. $\exists xK_{x}N(x,a)$ is a formula\textemdash read,
in this paper, as ``there there exists an agent that knows of itself
that it is a social network neighbor of $a$''. The \emph{dynamic}
term-modal logic of \cite{ALR} extends term-modal logic with suitably
generalized action models that can effectuate both factual changes
(e.g. to the network structure) as well as epistemic changes. For
all the $\mathsf{DTML}$ action model encodable dynamics, \cite{ALR} presents
a general sound and complete reduction axiom-based logic in the style
of \emph{dynamic epistemic logic} (DEL, \cite{BaltagBMS_1998,baltagmoss2004}). Hence, whenever
an epistemic social network dynamics is encodable using $\mathsf{DTML}$, completeness
follows. 
With this in mind, the main goal of this paper is to introduce and illustrate $\mathsf{DTML}$ as a formalism for representing epistemic social network dynamics, and to show how it may be used to obtain completeness results.

To this end, the paper progresses as follows. Sec. \ref{sec:DTML}
sketches some common themes in the logical literature on social
networks before introducing $\mathsf{DTML}$ and its application to
epistemic social networks. Sec. \ref{sec:DTML} contains the bulk
of the paper, with numerous examples of both static $\mathsf{DTML}$ models and action models. The examples are both meant to showcase the scope
of $\mathsf{DTML}$ and to explain the more non-standard technical
details involved in calculating updated models. In Sec. \ref{sec:hard.results},
we turn to technical results, where it is shown that $\mathsf{DTML}$
may encode popular static hybrid logical models of epistemic networks, as well as  
the dynamics of \cite{ChristoffHansenProietti2016}, for which finding a complete logic was left open.
Sec. \ref{sec:final} contains final remarks.

\section{Models and Languages for Epistemic Social Networks\label{sec:DTML}}

To situate $\mathsf{DTML}$ in the logical literature on social networks, we cannot
but describe the literature in broad terms. We omit both focus, formal
details and main results of the individual contributions in favor
of a broad perspective. That said, then all relevant literature in
one way or other concern \emph{social networks}. In general, a \textbf{social
network} is a graph $(A,N)$ where $A$ is a set of agents and $N\subseteq A\times A$
is represents a social relation, e.g., being friends on some social
media platform. Depending on interpretation, $N$ may be assumed irreflexive
and symmetric. Social networks may be augmented with assignments of
atomic properties to agents, representing e.g. behaviors, opinions
or beliefs. One set of papers investigates such models and their dynamics
using fully propositional static languages \cite{Rendsvig2014,ChristoffNaumov2018,SmetsVelazquez2017,SmetsVelazquez2018a}.

A second set of papers combines social networks with a semantically
represented epistemic dimension in the style of epistemic logic. In
these works, the fundamental structure of interest is (akin to) a
tuple 
\[
(A,W,\{N_{w}\}_{w\in W},\sim)
\]
with agents $A$ and worlds $W$, with each world $w$ associated
with a network $N_{w}\subseteq A\times A$, and $\sim:A\rightarrow\mathcal{P}(W\times W)$
associating each agent with an indistinguishability (equivalence) relation
$\sim_{a}$. Call such a tuple an \textbf{epistemic network structure}.

The existing work on epistemic network structures may be organized
in terms of the static languages they work with: \emph{propositional
modal logic} \cite{Baltag_etal_2018,ChristoffRendsvig2014} or \emph{hybrid
logic} \cite{Liang2011,Liu2014,Seligman2013,Seligman2011,Christoff2016,ChristoffHansen2013,ChristoffHansen2015,ChristoffHansenProietti2016}.
In the former, the social network is described using designated atomic
propositions (e.g., $N_{ab}$ for `$b$ is a neighbor of $a$').
To produce a model, an epistemic network structure is augmented with
a propositional valuation $V:P\to\mathcal{P}(W)$. Semantically, $N_{ab}$
is then true at $w$ iff $(a,b)\in N_{w}$.%
{} Knowledge is expressed using operators $\{K_{a}\}_{a\in A}$ as
in standard epistemic logic with $K_{a}$ the Kripke modality for
$\sim_{a}$.

In the hybrid case, the network is instead described using modal operators.
The hybrid languages typically include a set of agent nominals $Nom$
(agent names), atoms $P$ and \emph{indexical} modal operators
$K$ and $N$, read ``I know that'' and ``all my neighbors''. Some
papers additionally include state nominals, hybrid operators ($@_{x}$,
$\downarrow_{x}$) and/or universal modalities $U$ (``for all agents''). A \textbf{hybrid network
model} is an epistemic network structures extended with two assignments:
a nominal assignment $g:Nom\to A$ that names agents, and a two-dimensional
hybrid valuation $V:P\to \mathcal{P}(W\times A)$, where $(w,a)\in V(p)$ represents
that the indexical proposition $p$ holds of agent $a$ at $w$. The
satisfaction relation is relative to \textit{both} an epistemic alternative
$w$ and an agent $a$, where the noteworthy clause are: $M,w,a\models p$
iff $(w,a)\in V(p)$; $M,w,a\models K\varphi$ iff $M,v,a\models\varphi$
for every $v\sim_{a}w$; and $M,w,a\models N\varphi$ iff $M,w,b\models\varphi$
for every $b$ such that $N_w(a,b)$. With these semantics, formulas are read
indexically. E.g. $KNp$ reads as ``I know that all my neighbors are $p$''.

In relation to these two language types, the term-modal approach of
this paper lies closer to the former: By including a binary `neighbor
of' relation symbol $N$ in the signature of a term-modal
language, the social network component of models is described non-modally.
This straightforwardly allows expressing e.g. that that all agents
know all their neighbors ($\forall x\forall y(N(x,y)\rightarrow K_{x}(N(x,y))$)
or that an agent has \emph{de re} vs. \emph{de dicto }knowledge of
someone being a neighbor ($\exists xK_{\underline{a}}N(\underline{a},x)$ vs. $K_{\underline{a}}\exists xN(\underline{a},x)$). Moreover, hybrid languages can be translated into $\mathsf{DTML}$, in such a way that hybrid formulas such as $@_a p$ (``agent $a$ has property $p$'') become equivalent to $P(\underline{a})$, if $\underline{a}$ is the name of $a$.

\subsection{Term-Modal Logic and Epistemic Network Structures}

In general, term-modal languages may be based on any first-order signature,
by for the purposes of representing social networks and factual properties
of agents, we limit attention to the following:\footnote{The defined are special cases of the setting in \cite{ALR}, which allows general signatures and non-agent terms. \cite{ALR} also reviews the term-modal literature.}
\begin{definition}
\emph{A }\textbf{\emph{signature }}\emph{is a tuple $\Sigma=(\mathtt{V},\mathtt{C},\mathtt{P},N,\doteq)$
with $\mathtt{V}$ a countably infinite set of variables, $\mathtt{C}$
and $\mathtt{P}$ countable sets of constants and unary predicates,
$N$ a binary relation symbol and $\doteq$ for identity. The }\textbf{\emph{terms}}\emph{
of $\Sigma$ are $\mathtt{T}\coloneqq\mathtt{V}\cup\mathtt{C}$. With
$t_{1},t_{2}\in\mathtt{T},x\in\mathtt{V}$ and $P\in\mathtt{P}$,
the }\textbf{\emph{language}}\emph{ $\mathcal{L}(\Sigma)$ is given
by 
\[
\varphi\coloneqq P(t_{1})\mid N(t_{1},t_{2})\mid(t_{1}\doteq t_{2})\mid\neg\varphi\mid\varphi\wedge\varphi\mid K_{t}\varphi\mid\forall x\varphi
\]
Standard Boolean connectives, $\top$, $\exists$ and $\hat{K}_{t}$
are defined per usual. With $\varphi\in\mathcal{L},t\in\mathtt{T}$,
$x\in\mathtt{V}$, the result of replacing all occurrences of $x$
in $\varphi$ with $t$ is denoted $\varphi(x\mapsto t)$. Formulas
from the first three clauses are called }\textbf{\emph{atoms}}\emph{;
if an atom contains no variables, it is }\textbf{\emph{ground}}\emph{.}
\end{definition}

Throughout, $\underline{a},\underline{b}$, etc. are used for constants
and the relation symbol $N$ denotes a social network. The reading
of $N(t_{1},t_{2})$ depends on application. $K_{t}\varphi$ is a
term-indexed epistemic operator which read as ``agent $t$ knows
that $\varphi$''. $\mathcal{L}(\Sigma)$ neither enforces nor requires
a fixed-size agent set $A$, in contrast with standard epistemic languages,
where the set of operators is given by reference to $A$. Hence the
same language may be used to describe networks of varying size.

To interpret $\mathcal{L}(\Sigma)$, we use \emph{constant-domain}
models (the same number of agents in each world) with \emph{non-rigid}
constants (names, like predicates and relations, may change extension
between worlds; this allows for uncertainty about agent identity). See Figs. \ref{fig:static1} and \ref{ex2: pt1} for examples of such models.  

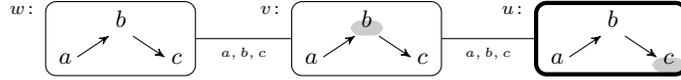
\begin{figure}
\begin{centering}
\begin{tikzpicture}[->,>=stealth',auto,shorten >=0pt,shorten <=0pt]
		\node[agent] (a0) {{\small$a$}};
		\node[agent, right=8pt of a0, yshift=14pt] (b) {{\small$b$}};
		\node[agent, right=8pt of b, yshift=-14pt] (c) {{\small$c$}};
		\draw (a0) -- (b);
		\draw (b) -- (c);
		\node[world, left=0pt of a0,xshift=56pt,yshift=10pt, label={[label distance=-2pt]168:\scriptsize{$w\colon$}}] (w1) {};

		\node[agent, right=80pt of a0] (a) {{\small$a$}};
		\node[agentpred, right=8pt of a, yshift=14pt] (b) {{\small$b$}};
		\node[agent, right=8pt of b, yshift=-14pt] (c) {{\small$c$}};
		\draw (a) -- (b);
		\draw (b) -- (c);
		\node[world, left=0pt of a,xshift=56pt,yshift=10pt, label={[label distance=-2pt]168:\scriptsize{$v\colon$}}] (w2) {};

		\node[agent, right=80pt of a] (a) {{\small$a$}};
		\node[agent, right=8pt of a, yshift=14pt] (b) {{\small$b$}};
		\node[agentpred, right=8pt of b, yshift=-14pt] (c) {{\small$c$}};
		\draw (a) -- (b);
		\draw (b) -- (c);
		\node[real, left=0pt of a,xshift=56pt,yshift=10pt, label={[label distance=-2pt]168:\scriptsize{$u\colon$}}] (w3) {};
	 	\draw[-] (w1) -- (w2) node[below,midway] {{\tiny $a,b,c$}};
	 	\draw[-] (w2) -- (w3) node[below,midway] {{\tiny $a,b,c$}};
\end{tikzpicture}
\par\end{centering}
\caption{\textbf{Example 1, pt. 1 (Server Error).}\label{fig:static1} Three
agents $a$,$b$ and $c$ work in a company with a hierarchical command
structure, $\protect\longrightarrow$: $a$ is the direct boss of
\textbf{$b$}, who is the direct boss of $c$.
The server has thrown
an error after both $b$ and $c$ tampered with it. Either $w)$ the
server failed spontaneously, $v)$ $b$ made a mistake (marked by
gray) or $u)$ $c$ made a mistake. Lines represent indistinguishability
with reflexive and transitive links omitted. There is no uncertainty
about the hierarchy, but nobody knows why the server failed. In fact,
$c$ made a mistake: the actual world has a thick outline.\vspace{-16pt}} 
\end{figure}

\begin{definition}
\emph{\label{model} An }$\mathcal{L}(\Sigma)$-\textbf{\emph{model}}\emph{
is a tuple $M=(A,W,\sim,I)$ where $A$ is a non-empty }\textbf{\emph{domain
of agents}}\emph{, $W$ is a non-empty set of }\textbf{\emph{worlds}}\emph{,
$\sim\,:A\rightarrow\mathcal{P}(W\times W)$ assigns to each agent
$a\in A$ an }\textbf{\emph{equivalence relation}}\emph{ on $W$ denoted
$\sim_{a}$, and $I$ is an }\textbf{\emph{interpretation}}\emph{
satisfying, for all $w\in W$, 1. for $c\in\mathtt{C}$, $I(c,w)\in A$;
2. for $P\in\mathtt{P}$, $I(P,w)\subseteq A$; 3. $I(N,w)\subseteq A\times A$}%
\emph{. A }\textbf{\emph{pointed model}}\emph{ is a pair $(M,w)$
with $w\in W$ called the }\textbf{\emph{actual world}}\emph{.}\\
\emph{A variable }\textbf{\emph{valuation }}\emph{of $\Sigma$ over
$M$ is a map $g:\mathtt{V}\rightarrow A$.}%
\emph{ The valuation identical to $g$ except mapping $x$ to $a$
is denoted }\textup{\emph{$g[x\mapsto a]$. }}\emph{The }\textbf{\emph{extension
}}\emph{of the term $t\in\mathtt{T}$ at $w$ in $M$ under $g$ is
$\left\llbracket t\right\rrbracket _{w}^{I,g}=g(t)$ for $t\in\mathtt{V}$
and $\left\llbracket t\right\rrbracket _{w}^{I,g}=I(t,w)$ for $t\in\mathtt{C}$.}
\end{definition}

\noindent Given the inclusion of $N$ in the signature $\Sigma$,
each $\mathcal{L}(\Sigma)$-model embeds an epistemic network structure
\textsc{$(A,W,(\sim_{a})_{a\in A},(I(N,w))_{w\in W})$.}

Formulas are evaluated over pointed models using a direct combination
of first-order and modal semantics:
\begin{definition}
\label{def:satisfaction}\emph{Let $\Sigma$, $M$ and $g$ be given.
The }\textbf{\emph{satisfaction}}\emph{ of formulas of $\mathcal{L}(\Sigma)$
is given recursively by}

\emph{$M,w\vDash_{g}P(t_{1})$ iff $\left\llbracket t_{1}\right\rrbracket _{w}^{I,g}\in I(P,w)$,
for $P\in\mathtt{P}$.}

\emph{$M,w\vDash_{g}N(t_{1},t_{2})$ iff $(\left\llbracket t_{1}\right\rrbracket _{w}^{I,g},\left\llbracket t_{2}\right\rrbracket _{w}^{I,g})\in I(N,w)$.}

\emph{$M,w\vDash_{g}(t_{1}\doteq t_{2})$ iff $\left\llbracket t_{1}\right\rrbracket _{w}^{I,g}=\left\llbracket t_{2}\right\rrbracket _{w}^{I,g}$.}

\emph{$M,w\vDash_{g}\neg\varphi$ iff not $M,w\vDash_{g}\varphi$.}

\emph{$M,w\vDash_{g}\varphi\wedge\psi$ iff $M,w\vDash_{g}\varphi$
and $M,w\vDash_{g}\psi$.}

\emph{$M,w\models_{g}\forall x\varphi$ iff $M,w\models_{g[x\mapsto a]}\varphi$
for all $a\in A$.}

\emph{$M,w\vDash_{g}K_{t}\varphi$ iff $M,w'\vDash_{g}\varphi$ for
all $w'$ such that $w\sim_{\left\llbracket t\right\rrbracket _{w}^{I,g}}w'$.}
\end{definition}
\pagebreak
\subsection{Knowing Who and Knowledge \emph{De Dicto} and \emph{De Re}}

First-order modal languages can represent propositional attitudes
\emph{de dicto }(about the statement) and \emph{de re} (about the
thing) in principled manners. For example, $K_{\underline{a}}\exists xP(x)$
is a \emph{de dicto} statement: knowledge is expressed about the proposition
that a $P$-thing exists. In contrast, $\exists xK_{\underline{a}}P(x)$
is a \emph{de re} statement: it is expressed that of some thing $x$,
that $x$ is known to be a $P$-thing. In general, \emph{de re} statements
are stronger than \emph{de dicto} statements. The difference has been
appreciated in epistemic logic since Hintikka's seminal \cite{TML_Hintikka1962},
where he argues that $\exists xK_{\underline{a}}(x\doteq\underline{b})$
expresses that $a$ \emph{knows who $b$ is} (see Fig. \ref{ex2: pt1}).
Semantically, the formula ensures that the constant $\underline{b}$
refers to the same individual in all $\underline{a}$'s epistemic
alternatives (i.e., $\underline{b}$ is \emph{locally rigid}). 
Both \emph{de dicto} and \emph{de re} statements may partially be expressed in propositional
languages (e.g. \emph{de dicto} $K_{a}(p_{b}\vee p_{c})$ vs. \emph{de
re} $K_{a}p_{b}\vee K_{a}p_{c}$; see \cite{Baltag_etal_2018} for
such a usage), but not in a principled manner: the required formulas
will depend on the specific circumstances.\vspace{-14pt}
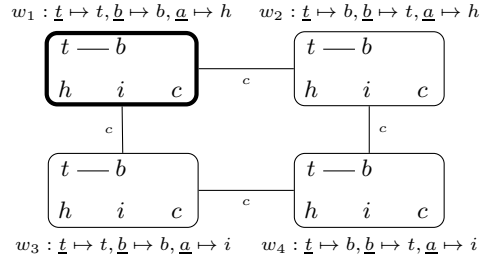
\begin{figure}[H]
\begin{centering}
\begin{tikzpicture}
	\node[agent] (t0) {{\small$t$}};
	\node[agent, right=8pt of t0] (b) {{\small$b$}};
	\node[agent, below=10pt of t0] (a) {{\small$h$}};
	\node[agent, right=8pt of a] (f) {{\small$i$}};
	\node[agent, right=8pt of f] (c) {{\small$c$}};
	\draw[shorten <=-2pt,shorten >=-2pt] (t0.20) -- (b.160); 
	\node[real, left=0pt of t0,xshift=57pt,yshift=-5pt, label={[label distance=0pt]90:\scriptsize{$w_1: \underline{t}\mapsto t, \underline{b} \mapsto b, \underline{a} \mapsto h$}}] (w1) {};

	\node[agent, right= 80pt of t0] (t1) {{\small$t$}};
	\node[agent, right=8pt of t1] (b) {{\small$b$}};
	\node[agent, below=10pt of t1] (a) {{\small$h$}};
	\node[agent, right=8pt of a] (f) {{\small$i$}};
	\node[agent, right=8pt of f] (c) {{\small$c$}};
	\draw[shorten <=-2pt,shorten >=-2pt] (t1.20) -- (b.160); 
	\node[world, left=0pt of t1, xshift=57pt,yshift=-5pt, label={[label distance=0pt]90:\scriptsize{$w_2: \underline{t}\mapsto b, \underline{b} \mapsto t, \underline{a} \mapsto h$}}] (w2) {};
	
	\node[agent, below= 40pt of t0] (t2) {{\small$t$}};
	\node[agent, right=8pt of t2] (b) {{\small$b$}};
	\node[agent, below=10pt of t2] (a) {{\small$h$}};
	\node[agent, right=8pt of a] (f) {{\small$i$}};
	\node[agent, right=8pt of f] (c) {{\small$c$}};
	\draw[shorten <=-2pt,shorten >=-2pt] (t2.20) -- (b.160); 
	\node[world, left=0pt of t2,xshift=57pt,yshift=-5pt, label={[label distance=-42pt]90:\scriptsize{$w_3: \underline{t}\mapsto t, \underline{b} \mapsto b, \underline{a} \mapsto i$}}] (w3) {};

	\node[agent, below= 40pt of t1] (t3) {{\small$t$}};
	\node[agent, right=8pt of t3] (b) {{\small$b$}};
	\node[agent, below=10pt of t3] (a) {{\small$h$}};
	\node[agent, right=8pt of a] (f) {{\small$i$}};
	\node[agent, right=8pt of f] (c) {{\small$c$}};
	\draw[shorten <=-2pt,shorten >=-2pt] (t3.20) -- (b.160); 
	\node[world, left=0pt of t3,xshift=57pt,yshift=-5pt, label={[label distance=-42pt]90:\scriptsize{$w_4: \underline{t}\mapsto b, \underline{b} \mapsto t, \underline{a} \mapsto i$}}] (w4) {};	
	
	 \draw[-] (w1) -- (w2) node[below,midway] {{\tiny $c$}};
	 \draw[-] (w2) -- (w4) node[right,midway] {{\tiny $c$}};
	 \draw[-] (w1) -- (w3) node[left,midway] {{\tiny $c$}};
	 \draw[-] (w3) -- (w4) node[below,midway] {{\tiny $c$}};		
\end{tikzpicture}%
\par\end{centering}
\centering{}\caption{\textbf{Example 2, pt.1 (Knowing Who).}\label{ex2: pt1} Two thieves,
$t$ and $b$, hide in a building with hostages $h$ and $i$. Outside,
a cop, $c$, waits. To communicate safely, the thieves use \textit{code
names} `Tokyo' and `Berlin' for each other and `The Asset' for
the specially valuable hostage $h$. Agents $t,b,h$ and $i$ all know whom
the code names denote (the names are rigid for them), but the cop
does not. The \textit{\emph{code names }}are $\underline{t}$ for
$t$, $\underline{b}$ for $b$ and $\underline{a}$ for $h$. Known
by all, $h$ and $i$ are in fact called $\underline{h}$ and $\underline{i}$. The thief network (\textemdash ) is assumed symmetric and transitive.
The case is modeled using four worlds, identical up to code name denotation,
(shown by $\protect\mapsto$). E.g., in the actual world is $w_{1}$,
$\underline{t}$ names $t$, but in $w_{4}$, it names $b$. Hence
the cop does not know who Tokyo is: $M,w_{1}\vDash_{g}\neg\exists xK_{\underline{c}}(x\protect\doteq\underline{t})$.}
\end{figure}

\subsection{Dynamics: Action Models and Product Update}

To code operations on static models, we use a a variant of DEL-style
action models, adapted to term-modal logic (see Fig. \ref{fig:event1}).
They include (adapted versions of) \textit{preconditions} specifying
when an event is executable (\cite{BaltagBMS_1998,baltagmoss2004}),
\textit{postconditions} describing the factual effects of events (\cite{ditmarsch_kooi_ontic,benthem2006_com-change,bolanderbirkegaard2011})
as well as \textit{edge-conditions} representing how an agent's observation
of an action depends on the agent's circumstances (\cite{Bolander2014})\textemdash for
example their position in a network, cf. Fig. \ref{fig:event1}. Edge-conditions
are non-standard and deserve a remark. With $E$ the set of events,
edge-conditions are assigned by a map $Q$. For each edge $(e,e')\in E\times E$,
$Q(e,e')$ is a formula with a single free variable $x^{\star}$.
Given a model $M$, an agent $i$ cannot distinguish $e$ from $e'$
iff the edge-condition $Q(e,e')$ is true in $M$\emph{ }\textit{\emph{when
the free variable $x^{\star}$ is mapped to $i$}}. Intuitively, if
the situation described by the edge-condition is true for $i$, the
way in which $i$ is observing the action does not allow her to tell
whether $e$ or $e'$ is taking place. See Figure \ref{fig:static2} for
an example. See \cite{ALR} for a comparison of this approach to
that of \cite{Bolander2014} and the term-modal action models of \cite{TML_KooiTermModalLogic}.
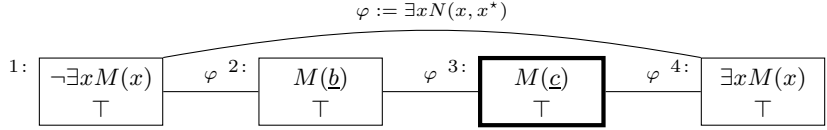
\begin{figure}[h]
\begin{centering}
\begin{tikzpicture}[->,>=stealth',auto,shorten >=0pt,shorten <=0pt,label distance=-2pt]]
	\node[event, text width = 40pt, label={170:\scriptsize{$1\colon$}}] (e1) {{\small$\neg\exists x M(x)$\\$\top$}};
	\node[event, text width = 40pt, right=36pt of e1, label={170:\scriptsize{$2\colon$}}] (e2) {{\small$M(\underline{b})$\\$\top$}};
	\node[realevent, text width = 40pt, right=36pt of e2, label={170:\scriptsize{$3\colon$}}] (e3) {{\small$M(\underline{c})$\\$\top$}};
	\node[event, text width = 40pt, right=36pt of e3, label={170:\scriptsize{$4\colon$}}] (e4) {{\small$\exists x M(x)$\\$\top$}};
	
	\draw[-] (e1) -- (e2) node[above,midway] {{\scriptsize $\varphi$}};
	\draw[-] (e2) -- (e3) node[above,midway] {{\scriptsize $\varphi$}};
	\draw[-] (e3) -- (e4) node[above,midway] {{\scriptsize $\varphi$}};
	\draw[-] (e1.north east) to [bend left=10] node[above,midway] {{\scriptsize $\varphi:= \exists x N(x,x^\star)$}} (e4.north west);
\end{tikzpicture}
\par\end{centering}
\caption{\textbf{Example 1, pt. 2 (Edge-Conditions: Announcement to Subgroup).}\label{fig:event1}
To learn what happened to the server, the top boss $a$ requests its
log file. The log holds one of four pieces of information: $1)$ Nobody
made a mistake, $2)$ $b$ made a mistake ($M$), $3)$ $c$ made
a mistake or $4)$ somebody made a mistake. Each box represents one
of these events: top lines are preconditions, bottom lines postconditions
($\top$ means no factual change). In fact, the log rats on $c$.
$N$ denotes the hierarchy. The log is send only to the top boss: the others cannot see its content.
This is represented by the edge-condition $\varphi$: If you, $x^{\star}$,
have a boss, then you cannot tell $1)$ from $2)$ nor $2)$ from
$3)$ etc. For unillustrated edges, $Q(e,e)=(x^{\star}\protect\doteq x^{\star})$
and $Q(e,e')=\varphi$ when $e\protect\neq e'$.\vspace{-12pt}}
\end{figure}

For simplicity, we here only define action models that take pre-,
post, and edge-conditions in the static language $\mathcal{L}(\Sigma)$.
However, dynamic conditions are needed for completeness; we refer
to \cite{ALR} for details.
\begin{definition}
\textup{An }\textbf{\textup{action model}}\textup{ for $\mathcal{L}(\Sigma)$
is }a tuple $\Delta=(E,Q,\mathsf{pre},\mathsf{post})$ where
\begin{itemize}
\item \emph{$E$ is a non-empty, finite set of }\textbf{\emph{events}}\emph{.}
\item \emph{$Q:(E\times E)\to\mathcal{L}(\Sigma)$ where each }\textbf{\emph{edge-condition}}\emph{
$Q(e,e')$ has exactly one free variable $x^{\star}$.}
\item \emph{$\mathsf{pre}:E\to\mathcal{L}(\Sigma)$ where each }\textbf{\emph{precondition}}\emph{
}\textup{\emph{$\mathsf{pre}(e)$}}\emph{ has no free variables.}
\item \emph{$\mathsf{post}:E\to(\mathtt{GroundAtoms}(\mathcal{L}(\Sigma))\to\mathcal{L}(\Sigma))$
assigns to each $e\in E$ a }\textbf{\emph{postcondition}}\emph{ for
each ground atom.}\\
\emph{To preserve the meaning of equality, let $\mathsf{post}(e)(t\doteq t)=\top$
for all $e\in E$.}
\end{itemize}
\end{definition}

\noindent With no general restrictions on $Q$, to ensure that all
agents' indistinguishability relations continue to be equivalence
relations after updating, $Q$ must be chosen with care. Throughout,
we assume $Q(e,e)=(x^{\star}\doteq x^{\star})$ for all $e\in E$.
To update, \emph{product update} may be altered to fit the edge-condition
term-modal setting as below. Fig. \ref{fig:static2} illustrates the
product update of Figs. \ref{fig:static1} with \ref{fig:event1}. The use
of postconditions is illustrated in Figs. \ref{fig:event4+static5} and \ref{ex2: pt2-1}.

\pagebreak

\begin{definition}
\emph{\label{Def. Product update} Let $M=(A,W,\sim,I)$ and $\Delta=(E,Q,\mathsf{pre},\mathsf{post})$
be given. The }\textbf{\emph{product update}}\emph{ of $M$ and $\Delta$
is the model $M\otimes\Delta=(A',W',\sim',I')$ where}
\begin{enumerate}
\item \emph{$A'=A$}
\item \emph{$W'=\{(w,e)\in W\times E\colon(M,w)\vDash_{g}\mathsf{pre}(e)\}$
for any $g$,}
\item \emph{$(w,e)\sim'_{i}(w',e')$ iff $w\sim_{i}w'$ and $M,w\vDash_{g[x^{\star}\mapsto i]}Q(e,e')$,}
\item \emph{$I'(c,(w,e))=I(c,w)$ for all $c\in\mathtt{C}$, and}\\
\emph{$I'(X,(w,e))=(I(X,w)\cup X^{+}(w))\setminus X^{-}(w)$, for
$X=\{P,N\},P\in\mathtt{P},$ where: 
\begin{align*}
P^{+}(w)\coloneqq & \{\llbracket t\rrbracket_{w}^{I,v}\colon(M,w)\vDash_{g}\mathsf{post}(e)(P(t))\};\\
P^{-}(w)\coloneqq & \{\llbracket t\rrbracket_{w}^{I,v}\colon(M,w)\not\vDash_{g}\mathsf{post}(e)(P(t))\};\\
N^{+}(w)\coloneqq & \{(\llbracket t_{1}\rrbracket_{w}^{I,v},\llbracket t_{2}\rrbracket_{w}^{I,v})\colon(M,w)\vDash_{g}\mathsf{post}(e)(N(t_{1},t_{2}))\};\\
N^{-}(w)\coloneqq & \{(\llbracket t_{1}\rrbracket_{w}^{I,v},\llbracket t_{2}\rrbracket_{w}^{I,v})\colon(M,w)\not\vDash_{g}\mathsf{post}(e)(N(t_{1},t_{2}))\}
\end{align*}
}
\end{enumerate}
\emph{If $(M,w)\models\mathsf{pre}(e)$, then $(A,e)$ is }\textbf{\emph{applicable}}\emph{
to $(M,w)$, and the product update of the two is the pointed model
$(M\otimes\Delta,(w,e))$. Else it is undefined.}
\end{definition}

\noindent 
\begin{figure}[h]
\begin{centering}
\begin{tikzpicture}[->,>=stealth',auto,shorten >=0pt,shorten <=0pt]
		\node[agent] (a0) {{\small$a$}};
		\node[agent, right=8pt of a0, yshift=14pt] (b) {{\small$b$}};
		\node[agent, right=8pt of b, yshift=-14pt] (c) {{\small$c$}};
		\draw (a0) -- (b);
		\draw (b) -- (c);
		\node[world, left=0pt of a0,xshift=56pt,yshift=10pt, label={[label distance=-2pt]168:\scriptsize{$w1\colon$}}] (w1) {};
		\node[agent, right=80pt of a0, yshift=23pt] (a1) {{\small$a$}};
		\node[agentpred, right=8pt of a1, yshift=14pt] (b) {{\small$b$}};
		\node[agent, right=8pt of b, yshift=-14pt] (c) {{\small$c$}};
		\draw (a1) -- (b);
		\draw (b) -- (c);
		\node[world, left=0pt of a1,xshift=56pt,yshift=10pt, label={[label distance=-2pt]168:\scriptsize{$v2\colon$}}] (v2) {};

		\node[agent, right=80pt of a1] (a) {{\small$a$}};
		\node[agent, right=8pt of a, yshift=14pt] (b) {{\small$b$}};
		\node[agentpred, right=8pt of b, yshift=-14pt] (c) {{\small$c$}};
		\draw (a) -- (b);
		\draw (b) -- (c);
		\node[real, left=0pt of a,xshift=56pt,yshift=10pt, label={[label distance=-2pt]168:\scriptsize{$u3\colon$}}] (u3) {};

		\node[agent, below=40pt of a1] (a) {{\small$a$}};
		\node[agentpred, right=8pt of a, yshift=14pt] (b) {{\small$b$}};
		\node[agent, right=8pt of b, yshift=-14pt] (c) {{\small$c$}};
		\draw (a) -- (b);
		\draw (b) -- (c);
		\node[world, left=0pt of a,xshift=56pt,yshift=10pt, label={[label distance=-2pt]168:\scriptsize{$v4\colon$}}] (v4) {};

		\node[agent, right=80pt of a] (a) {{\small$a$}};
		\node[agent, right=8pt of a, yshift=14pt] (b) {{\small$b$}};
		\node[agentpred, right=8pt of b, yshift=-14pt] (c) {{\small$c$}};
		\draw (a) -- (b);
		\draw (b) -- (c);
		\node[world, left=0pt of a,xshift=56pt,yshift=10pt, label={[label distance=-2pt]168:\scriptsize{$u4\colon$}}] (u4) {};
	%
	 	\draw[-] (w1) -- (v2.180) node[below,midway] {{\tiny $b,c$}};
	 	\draw[-] (v2) -- (u3) node[below,midway] {{\tiny $b,c$}};
		\draw[-] (w1) -- (v4.180) node[below,midway] {{\tiny $b,c$}}; 	
		\draw[-] (v4) -- (u4) node[below,midway] {{\tiny $a,b,c$}};
	  	\draw[-] (v2) -- (v4) node[right,midway] {{\tiny $b,c$}};
	  	\draw[-] (u3) -- (u4) node[right,midway] {{\tiny $b,c$}};
\end{tikzpicture}\vspace{-8pt}
\par\end{centering}
\caption{\textbf{Example 1, pt. 3 (Product Update: Edge-Conditions).}\label{fig:static2}
The product update of Fig. \ref{fig:static1} and Fig. \ref{fig:event1}.
After checking the logs, the boss has learned that $c$ made a mistake,
while both $b$ and $c$ are now both uncertain about this, as well
as about the boss' information. Worlds are named using by the world-event
pair they represent: $w1$ is the child of $w$ and $1$, etc. The
pair $w2$ is not a world: $w$ did not satisfy the precondition of
$1$. We have $w1\sim'_{b}v2$ as $w\sim_{b}v$ and $M,w\vDash_{g[x^{\star}\protect\mapsto a]}Q(1,2)$\textemdash as
$M,w\vDash_{g}\exists xN(x,\underline{b})$. Likewise, $v2\sim'_{b}w1$
as $v\sim_{b}w$ and $M,v\vDash_{g}\exists xN(x,\underline{b})$.
That $w1\protect\not\sim'_{a}v2$ follows as $M,w\vDash_{g}\neg\exists xN(x,\underline{b})$,
but $v4\sim'_{a}u4$ as $M,v\vDash_{g}(\underline{a}\protect\doteq\underline{a})$.
The same reason, reflexive loops are preserved. The boss now knows
that $c$ made a mistake: $K_{\underline{a}}M(\underline{c})$.\vspace{-12pt}}
\end{figure}
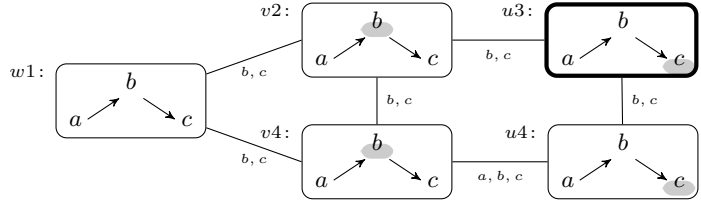

\subsection{Announcements \emph{De Dicto} and \emph{De Re}}

With \emph{de dicto} and \emph{de re} statements expressible in $\mathsf{DTML}$,
they may be used to define principled announcements, as exemplified
in Fig. \ref{fig:event2+static3.de.dicto} and \ref{fig:event3+static4.de.re}.
The action models are applicable to \emph{any} $\mathsf{DTML}$ model for a signature
that includes the constant $\underline{a}$ and the predicate $M$,
irrespective of the size of the set of agents. This level of general
applicability is not mirrored in standard DEL action models.\vspace{-12pt}\enlargethispage{2\baselineskip}
\begin{figure}[H]
\begin{centering}
\begin{tabular}{>{\centering}m{0.33\textwidth}|>{\centering}m{0.63\textwidth}}
\begin{tikzpicture}[->,>=stealth',auto,shorten >=0pt,shorten <=0pt]
	\node[realevent, label={[label distance=-2pt]168:\scriptsize{$e\colon$}}] (e1) {{\small$K_{\underline{a}}\exists x M(x)$\\ $\top$}};
	\end{tikzpicture} & \begin{tikzpicture}[->,>=stealth',auto,shorten >=0pt,shorten <=0pt]
		\node[agent] (a1) {{\small$a$}};
		\node[agentpred, right=8pt of a1, yshift=14pt] (b) {{\small$b$}};
		\node[agent, right=8pt of b, yshift=-14pt] (c) {{\small$c$}};
		\draw (a1) -- (b);
		\draw (b) -- (c);
		\node[world, left=0pt of a1,xshift=56pt,yshift=10pt, label={[label distance=-2pt]168:\scriptsize{$v2e\colon$}}] (v2) {};

		\node[agent, right=80pt of a1] (a) {{\small$a$}};
		\node[agent, right=8pt of a, yshift=14pt] (b) {{\small$b$}};
		\node[agentpred, right=8pt of b, yshift=-14pt] (c) {{\small$c$}};
		\draw (a) -- (b);
		\draw (b) -- (c);
		\node[real, left=0pt of a,xshift=56pt,yshift=10pt, label={[label distance=-2pt]168:\scriptsize{$u3e\colon$}}] (u3) {};

		\node[agent, below=40pt of a1] (a) {{\small$a$}};
		\node[agentpred, right=8pt of a, yshift=14pt] (b) {{\small$b$}};
		\node[agent, right=8pt of b, yshift=-14pt] (c) {{\small$c$}};
		\draw (a) -- (b);
		\draw (b) -- (c);
		\node[world, left=0pt of a,xshift=56pt,yshift=10pt, label={[label distance=-2pt]168:\scriptsize{$v4e\colon$}}] (v4) {};

		\node[agent, right=80pt of a] (a) {{\small$a$}};
		\node[agent, right=8pt of a, yshift=14pt] (b) {{\small$b$}};
		\node[agentpred, right=8pt of b, yshift=-14pt] (c) {{\small$c$}};
		\draw (a) -- (b);
		\draw (b) -- (c);
		\node[world, left=0pt of a,xshift=56pt,yshift=10pt, label={[label distance=-2pt]168:\scriptsize{$u4e\colon$}}] (u4) {};
	%
	 	\draw[-] (v2) -- (u3) node[below,midway] {{\tiny $b,c$}};
		\draw[-] (v4) -- (u4) node[below,midway] {{\tiny $a,b,c$}};
	  	\draw[-] (v2) -- (v4) node[right,midway] {{\tiny $b,c$}};
	  	\draw[-] (u3) -- (u4) node[right,midway] {{\tiny $b,c$}};
\end{tikzpicture}\tabularnewline
\end{tabular}\vspace{-8pt}
\par\end{centering}
\caption{\textbf{Example 1, pt. 3 (De Dicto Announcement).}\label{fig:event2+static3.de.dicto}
The boss breaks the news from the log to $b$ and $c$ piecemeal.
\textbf{\emph{Left:}} First, $a$ makes a \emph{de dicto }announcement:
$a$ knows that somebody made a mistake. \textbf{\emph{Right:}} The
effect on Fig. \ref{fig:static2}. Only $w1$ does not survive. In
$u3e$, everybody knows \emph{de dicto }that somebody messed up: $\forall xK_{x}\exists yM(y)$.
The boss also knows \emph{de re},\emph{ }i.e., knows \emph{who}: $u3e\vDash_{g}\exists xK_{a}M(x)$,
as $u3e\vDash_{g[x\protect\mapsto c]}K_{a}M(x)$. The employees do
not know that $a$ knows \emph{de re}: $u3e\vDash_{g}\forall x(\exists yN(y,x)\rightarrow\hat{K}_{x}\neg\exists zK_{a}M(z))$\textemdash since
$v4e\vDash_{g}M(x)$ iff $g(x)=b$, but then $u4e\protect\not\vDash_{g}M(x)$.
I.e., there is no \emph{one }object to serve as valuation for $x$
such that $v4e$ and $u4e$ satisfy $M(x)$ simultaneously). The employees
are held in suspense!}

\medskip

\begin{centering}
\vspace{8pt}%
\begin{tabular}{>{\centering}m{0.33\textwidth}|>{\centering}m{0.63\textwidth}}
\begin{tikzpicture}[->,>=stealth',auto,shorten >=0pt,shorten <=0pt]
	\node[realevent, label={[label distance=-2pt]168:\scriptsize{$\sigma\colon$}}] (e1) {{\small$\exists x K_{\underline{a}} M(x)$\\ $\top$}};
	\end{tikzpicture} & \begin{tikzpicture}[->,>=stealth',auto,shorten >=0pt,shorten <=0pt]
		\node[agent] (a1) {{\small$a$}};
		\node[agentpred, right=8pt of a1, yshift=14pt] (b) {{\small$b$}};
		\node[agent, right=8pt of b, yshift=-14pt] (c) {{\small$c$}};
		\draw (a1) -- (b);
		\draw (b) -- (c);
		\node[world, left=0pt of a1,xshift=56pt,yshift=10pt, label={[label distance=-2pt]168:\scriptsize{$v2e\sigma\colon$}}] (w2) {};

		\node[agent, right=80pt of a1] (a) {{\small$a$}};
		\node[agent, right=8pt of a, yshift=14pt] (b) {{\small$b$}};
		\node[agentpred, right=8pt of b, yshift=-14pt] (c) {{\small$c$}};
		\draw (a) -- (b);
		\draw (b) -- (c);
		\node[real, left=0pt of a,xshift=56pt,yshift=10pt, label={[label distance=-2pt]168:\scriptsize{$u3e\sigma\colon$}}] (w3) {};
	 	\draw[-] (w2) -- (w3) node[below,midway] {{\tiny $b,c$}};
\end{tikzpicture}\tabularnewline
\end{tabular}\vspace{-8pt}
\par\end{centering}
\caption{\textbf{Example 1, pt. 4 (De Re Announcement).}\label{fig:event3+static4.de.re}
Following a dramatic pause, the boss reveals a stronger piece of information:
the boss knows who messed up. This \emph{de re }announcement is on
the left, with $Q(e,e)=(x^{\star}=x^{\star})$; its result on Fig.
\ref{fig:event2+static3.de.dicto} (Right) on the right. In $u3e\sigma$,
everybody knows that $a$ has \emph{de re} knowledge: $\forall xK_{x}\exists yK_{a}M(y)$,
but $b$ and $c$ still only have \emph{de dicto} knowledge: $\forall x((x=b\vee x=c)\rightarrow K_{x}\exists yM(y)\wedge\neg\exists zK_{x}M(z))$.}
\end{figure}

\subsection{Postconditions and Network Change}\vspace{-6pt}
Action models with postconditions allows $\mathsf{DTML}$ to represent changes
to the social network. Such changes may be combined with the general
functionality of action models such that some agents may know what
changes occur while others remain in the dark. Fig. \ref{fig:event4+static5}
provides a simple example, including the details calculating the updated
network. Fig. \ref{ex2: pt2-1} presents an example of how \emph{de re/de dicto
} knowledge affects what is learned by a publicly observed network change.%
\vspace{-18pt}
\begin{figure}[H]
\begin{centering}
\begin{tabular}{>{\centering}m{0.48\textwidth}|>{\centering}m{0.48\textwidth}}
\begin{tikzpicture}[->,>=stealth',auto,shorten >=0pt,shorten <=0pt]
    \node[realevent, text width=142pt, label={167:\scriptsize{$\dagger\colon$}}]
        (e1) {{\small$\top$\\ $N(a,b),N(b,c)\mapsto\bot, N(a,c)\mapsto\top$}}; 	
\end{tikzpicture} & \begin{tikzpicture}[->,>=stealth',auto,shorten >=0pt,shorten <=0pt]
		\node[agent] (a1) {{\small$a$}};
		\node[agentpred, right=8pt of a1, yshift=14pt] (b) {{\small$b$}};
		\node[agent, right=8pt of b, yshift=-14pt] (c) {{\small$c$}};
		\draw (a1.north east) -- (c.north west);
		\node[world, left=0pt of a1,xshift=56pt,yshift=10pt, label={[label distance=-2pt]100:\scriptsize{$v2e\sigma\dagger\colon$}}] (w2) {};

		\node[agent, right=80pt of a1] (a) {{\small$a$}};
		\node[agent, right=8pt of a, yshift=14pt] (b) {{\small$b$}};
		\node[agentpred, right=8pt of b, yshift=-14pt] (c) {{\small$c$}};
		\draw (a.north east) -- (c.north west);
		\node[real, left=0pt of a,xshift=56pt,yshift=10pt, label={[label distance=-2pt]100:\scriptsize{$u3e\sigma\dagger\colon$}}] (w3) {};
	 	\draw[-] (w2) -- (w3) node[above,midway] {{\tiny $b,c$}};
\end{tikzpicture}\tabularnewline
\end{tabular}
\par\end{centering}
\caption{\textbf{Example 1, pt. 5 (Getting Fired).}\label{fig:event4+static5}
The employees are dying to know who messed up the server. But the
boss just proclaims: `$b$, you are fired! $c$, you are promoted!'
\textbf{\emph{Left: }}Action with three instructions for factual change:
$\mathsf{post}(\dagger)(N(a,b))=\bot$, $\mathsf{post}(\dagger)(N(b,c))=\bot$
and $\mathsf{post}(\dagger)(N(a,c))=\top$ (illustrated by $\protect\mapsto$).
Else $\mathsf{post}=id$. As $u3e\sigma\protect\not\vDash\bot$, the
first two instructions entail that $(a,b),(b,c)\in N^{-}(u3e\sigma)$,
while the latter implies that $(a,c)\in N^{+}(u3e\sigma)$. \textbf{\emph{Right:}}
The network is updated to $I'(N,u3e\sigma\dagger)=(I(N,u3e\sigma)\cup N^{+}(u3e\sigma))\backslash N^{-}(u3e\sigma)=(\{(a,b),(b,c)\}\cup\{(a,c)\})\backslash\{(a,b),(b,c)\}=\{(a,b)\}$.
In $u3e\sigma\dagger$, neither $b$ nor $c$ know who made the mistake.
Unrepresented, $a$ thinks that only bad superiors let their employees
make mistakes.}
\end{figure}
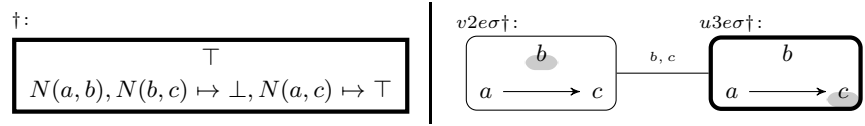
\vspace{-42pt}
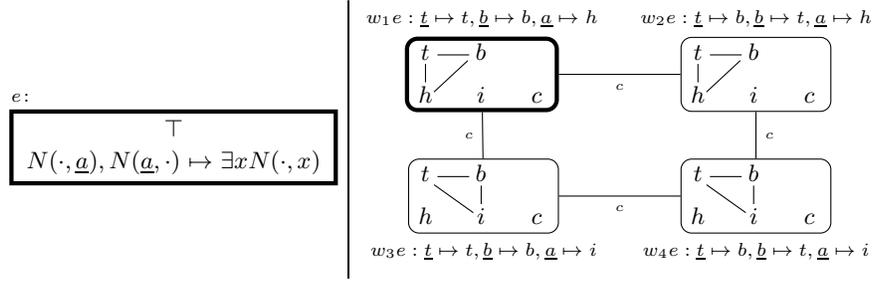
\begin{figure}[H]
\centering{}%
\begin{tabular}{>{\centering}m{0.38\textwidth}|>{\centering}m{0.58\textwidth}}
\begin{tikzpicture}[label distance=-2pt]]
	\node[realevent, text width = 116pt, label={164:\scriptsize{$e\colon$}}] (e1) {{\small$\top$\\$N(\cdot,\underline{a}), N(\underline{a},\cdot) \mapsto \exists x N(\cdot,x)$}};
\end{tikzpicture} & \begin{tikzpicture}
	\node[agent] (t0) {{\small$t$}};
	\node[agent, right=8pt of t0] (b) {{\small$b$}};
	\node[agent, below=10pt of t0] (a) {{\small$h$}};
	\node[agent, right=8pt of a] (f) {{\small$i$}};
	\node[agent, right=8pt of f] (c) {{\small$c$}};
	\draw[shorten <=-2pt,shorten >=5pt] (t0) -- (a);
	\draw[shorten <=-2pt,shorten >=-2pt] (t0.20) -- (b.160); 
	\draw[shorten >=2pt,shorten <=-2pt] (b.190) -- (a.60);
	\node[real, left=0pt of t0,xshift=57pt,yshift=-5pt, label={[label distance=0pt]90:\scriptsize{$w_1e: \underline{t}\mapsto t, \underline{b} \mapsto b, \underline{a} \mapsto h$}}] (w1) {};

	\node[agent, right= 90pt of t0] (t1) {{\small$t$}};
	\node[agent, right=8pt of t1] (b) {{\small$b$}};
	\node[agent, below=10pt of t1] (a) {{\small$h$}};
	\node[agent, right=8pt of a] (f) {{\small$i$}};
	\node[agent, right=8pt of f] (c) {{\small$c$}};
	\draw[shorten <=-2pt,shorten >=5pt] (t1) -- (a);
	\draw[shorten <=-2pt,shorten >=-2pt] (t1.20) -- (b.160); 
	\draw[shorten >=2pt,shorten <=-2pt] (b.190) -- (a.60);
	\node[world, left=0pt of t1, xshift=57pt,yshift=-5pt, label={[label distance=0pt]90:\scriptsize{$w_2e: \underline{t}\mapsto b, \underline{b} \mapsto t, \underline{a} \mapsto h$}}] (w2) {};
	
	\node[agent, below= 40pt of t0] (t2) {{\small$t$}};
	\node[agent, right=8pt of t2] (b) {{\small$b$}};
	\node[agent, below=10pt of t2] (a) {{\small$h$}};
	\node[agent, right=8pt of a] (f) {{\small$i$}};
	\node[agent, right=8pt of f] (c) {{\small$c$}};
	\draw[shorten <=-2pt,shorten >=5pt] (b) -- (f);
	\draw[shorten <=-2pt,shorten >=-2pt] (t2.20) -- (b.160); 
	\draw[shorten >=3pt,shorten <=-4pt] (t2.-25) -- (f.100);
	\node[world, left=0pt of t2,xshift=57pt,yshift=-5pt, label={[label distance=-42pt]90:\scriptsize{$w_3e: \underline{t}\mapsto t, \underline{b} \mapsto b, \underline{a} \mapsto i$}}] (w3) {};

	\node[agent, below= 40pt of t1] (t3) {{\small$t$}};
	\node[agent, right=8pt of t3] (b) {{\small$b$}};
	\node[agent, below=10pt of t3] (a) {{\small$h$}};
	\node[agent, right=8pt of a] (f) {{\small$i$}};
	\node[agent, right=8pt of f] (c) {{\small$c$}};
	\draw[shorten <=-2pt,shorten >=5pt] (b) -- (f);
	\draw[shorten <=-2pt,shorten >=-2pt] (t3.20) -- (b.160); 
	\draw[shorten >=3pt,shorten <=-4pt] (t3.-25) -- (f.100);
	\node[world, left=0pt of t3,xshift=57pt,yshift=-5pt, label={[label distance=-42pt]90:\scriptsize{$w_4e: \underline{t}\mapsto b, \underline{b} \mapsto t, \underline{a} \mapsto i$}}] (w4) {};	
	
	 \draw[-] (w1) -- (w2) node[below,midway] {{\tiny $c$}};
	 \draw[-] (w2) -- (w4) node[right,midway] {{\tiny $c$}};
	 \draw[-] (w1) -- (w3) node[left,midway] {{\tiny $c$}};
	 \draw[-] (w3) -- (w4) node[below,midway] {{\tiny $c$}};		
\end{tikzpicture}\tabularnewline
\end{tabular}\caption{\textbf{Example 2, pt.2 (Becoming Criminal)}
\textbf{\emph{Left:}} The thieves convince The Asset to cooperate with
them, in exchange for stolen goods. For simplicity, assume that the action of $\underline{a}$ joining the thief network is noticed by everyone. We model
this with the action model, with $\mathsf{post}(e)(N(\cdot,\underline{a}))=\exists xN(\cdot,x)$
and $\mathsf{post}(e)(N(\underline{a},\cdot))=\exists xN(x,\cdot)$
for $\cdot\in\{\underline{t},\underline{b},\underline{a},\underline{h},\underline{i},\underline{c}\}$.
Informally, these say: ``If you are a member of the network, then
$\underline{a}$ becomes your neighbor''.
\textbf{\emph{
Right:}}\textbf{ }The effect of event $e$ on Fig. \ref{ex2: pt1}: The
network has changed in all worlds, but differently. E.g., in $w_{1}$,
we had $\neg N(\underline{b},\underline{a})$; in $(w_{1},e)$, we have
$N(\underline{b},\underline{a})$ as $(b,h)\in N^{+}((w_{1},e))$
since $w_{1}\vDash_{g}\mathsf{post}(e)(N(\underline{b},\underline{a}))$\textemdash i.e.,
$\exists xN(\underline{b},x)$. Now all thieves and hostages know
the new network, as they know whom $\underline{a}$ refers to. E.g.: Tokyo knows all her neighbors,
$(w_{1},e)\vDash_{g}\forall x(N(\underline{t},x)\to K_{\underline{t}}N(\underline{t},x))$. The cop only learns that \textit{some}
hostage has joined the network, but can't tell whom: $(w_{1},e)\vDash_{g}K_{\underline{c}}\exists x(x\protect\not\protect\doteq\underline{t}\wedge x\protect\not\protect\doteq\underline{b}\wedge N(\underline{t},x))$
but $(w_{1},e)\protect\not\vDash_{g}\exists xK_{\underline{c}}(x\protect\not\protect\doteq\underline{t}\wedge x\protect\not\protect\doteq\underline{b}\wedge N(\underline{t},x))$.}
\label{ex2: pt2-1}\label{ex2: pt3-1}
\end{figure}
\vspace{-24pt}
\subsection{Learning Who}\vspace{-4pt}

Allowing for the possibility of non-rigid names has the consequence
that public announcements of atomic propositions may differ in informational
content depending on the epistemic state of the listener. This can
be exploited by the thieves of Example 2 to enforce a form of \textit{privacy}\textit{\emph{\textemdash as
code names should}}. The notion of privacy involved is orthogonal
to the notion of privacy modeled in DEL using private announcements.
Though the message is public in the standard sense of everyone being
aware of it and its content, as it involves non-rigid names,
its epistemic effects are not the same for all agents. This is in
contrast with standard public announcements, which yield the same
information to everyone. \vspace{-14pt}

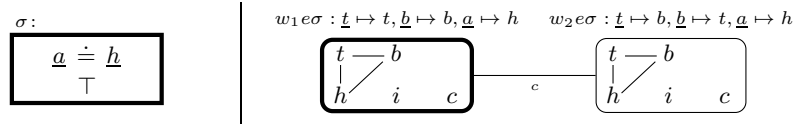
\begin{figure}[H]
\centering{}%
\begin{tabular}{>{\centering}m{0.33\textwidth}|>{\centering}m{0.63\textwidth}}
\begin{tikzpicture}[->,>=stealth',auto,shorten >=0pt,shorten <=0pt,label distance=-2pt]
	\node[realevent, label={136:\scriptsize{$\sigma\colon$}}] (e1) {{\small$\underline{a} \doteq \underline{h}$\\$\top$}};
\end{tikzpicture} & \centering{}\begin{tikzpicture}
	\node[agent] (t0) {{\small$t$}};
	\node[agent, right=8pt of t0] (b) {{\small$b$}};
	\node[agent, below=10pt of t0] (a) {{\small$h$}};
	\node[agent, right=8pt of a] (f) {{\small$i$}};
	\node[agent, right=8pt of f] (c) {{\small$c$}};
	\draw[shorten <=-2pt,shorten >=5pt] (t0) -- (a);
	\draw[shorten <=-2pt,shorten >=-2pt] (t0.20) -- (b.160); 
	\draw[shorten >=2pt,shorten <=-2pt] (b.190) -- (a.60);
	\node[real, left=0pt of t0,xshift=57pt,yshift=-5pt, label={[label distance=0pt]90:\scriptsize{$w_1e\sigma: \underline{t}\mapsto t, \underline{b} \mapsto b, \underline{a} \mapsto h$}}] (w1) {};

	\node[agent, right= 90pt of t0] (t1) {{\small$t$}};
	\node[agent, right=8pt of t1] (b) {{\small$b$}};
	\node[agent, below=10pt of t1] (a) {{\small$h$}};
	\node[agent, right=8pt of a] (f) {{\small$i$}};
	\node[agent, right=8pt of f] (c) {{\small$c$}};
	\draw[shorten <=-2pt,shorten >=5pt] (t1) -- (a);
	\draw[shorten <=-2pt,shorten >=-2pt] (t1.20) -- (b.160); 
	\draw[shorten >=2pt,shorten <=-2pt] (b.190) -- (a.60);
	\node[world, left=0pt of t1, xshift=57pt,yshift=-5pt, label={[label distance=0pt]90:\scriptsize{$w_2e\sigma: \underline{t}\mapsto b, \underline{b} \mapsto t, \underline{a} \mapsto h$}}] (w2) {};
	
	 \draw[-] (w1) -- (w2) node[below,midway] {{\tiny $c$}};
\end{tikzpicture}%
\tabularnewline
\end{tabular}\caption{\textbf{Example 2, pt.4 (Revealing the Asset)}\label{ex2: pt4} In
the model in Fig. \ref{ex2: pt3-1} (Right), even a \textit{public
announcement} of $N(\underline{t},\underline{a})$ would not inform
the cop about who joined the network. To know who joined the network,
the cop must learn \textit{who The Asset is}. As the cop knows who
$\underline{h}$ is, learning that $\underline{h}$ is The Asset suffices.
\textbf{\emph{Left: }}The event model $\sigma$ for the public announcement that $\underline{a}\protect\doteq\underline{h}$,
revealing the identity of The Asset. \textbf{\emph{Right: }}\label{ex2: pt5}
The product update of Fig. \ref{ex2: pt3-1} (Right) and event $\sigma$. The cop now knows the structure of the network, as a result of the removal
of $w_{3}e$ and $w_{4}e$.}
\end{figure}
\section{Embedding Dynamic Social Network Logics in \textsf{DTML}
\label{sec:hard.results}}

This section examines relations between the hybrid network models
and their languages to $\mathsf{DTML}$. As hybrid languages corresponds
to fragments of first-order logic with equality ($\mathsf{FOL}_{=}$),
which term-modal logic extends, it stands to reason that the hybrid
languages and models mentioned in Sec. 2 may be embedded in term-modal
logic. A precise statement and a proof sketch follows below. Turning
to dynamics, things are more complicated. \cite{Seligman2013} presents
a very flexible hybrid framework expressing network dynamics using
\emph{General Dynamic Dynamic Logic} ($\mathsf{GDDL}$, \cite{girard2012general}).
We leave general characterizations of equi-expressive fragments of
$\mathsf{GDDL}$ and $\mathsf{DTML}$ as open question, but remark
that all $\mathsf{GDDL}$ action-examples of \cite{Seligman2013}
may be emulated using $\mathsf{DTML}$ action models, and in many
cases via fairly simple ones. More thoroughly, we show that the logic
of \emph{Knowledge, Diffusion and Learning} ($\mathsf{KDL}$, \cite{ChristoffHansenProietti2016})
has a complete and decidable system, a question left open in \cite{ChristoffHansenProietti2016}.
This is shown by encoding $\mathsf{KDL}$ in $\mathsf{DTML}$.

\subsection{Embedding Static Languages and Models}

The static hybrid languages of  \cite{Liu2014,Seligman2013,Seligman2011,Christoff2016,ChristoffHansen2013,ChristoffHansen2015,ChristoffHansenProietti2016} are all sub-languages of
$\mathcal{L}(P,Nom)$, defined and translated into $\mathsf{DTML}$
below.  \cite{Liang2011} also includes state nominals, which our results
do not cover. $\mathcal{L}(P,Nom)$ is read indexically, as described in Sec. 2.

\begin{definition}
With $p\in P$ and $x\in Nom$, the \textbf{language} $\mathcal{L}(P,Nom)$
is given by 
\[
\varphi\coloneqq p\mid\neg\varphi\mid\varphi\wedge\varphi\mid@_{x}\varphi\mid K\varphi\mid N\varphi\mid U\varphi
\]
Denote the fragments without $U$ and $@_{x}$ by $\mathcal{L}_{-U}(P,Nom)$
and $\mathcal{L}_{-@}(P,Nom)$.
\end{definition}

Hybrid logics may be translated into $\mathsf{FOL}_{=}$; our translation resembles that of \cite{brauner2011hybrid}. We identify agent nominals
with first-order variables, translate the modal operator $N$ to the
relation symbol $N(\cdot,\cdot)$, and relativize the interpretation
of the indexical $K$ to the nominal/variable $x$ by using the term-indexed
operator $K_{x}$. Formally, the translation is defined as follows.
\begin{definition}
\label{def: translat} Let $\Sigma_{n}(P,Nom)=(\mathtt{V},\mathtt{C},\mathtt{P},N,\dot{=})$
be the signature with $\mathtt{V}=Nom$,
$\mathtt{C}=\{\underline{a}_{1},\dots,\underline{a}_{n}\}$ and $\mathtt{P}=P$.
Translations $T_{x},T_{y}$ both mapping $\mathcal{L}(P,Nom)$ to
\textup{$\mathcal{L}(\Sigma_{n}(P,Nom))$} are defined by mutual recursion. It is assumed that two nominals $x$ and $y$ are given which do not occur in the formulas to be translated. For $p\in P$ and $i\in Nom$, define $T_{x}$ by:
\begin{align*}
T_{x}(p) & =p(x) & T_{x}(@_{i}\varphi) & =T_{x}(\varphi)(x\mapsto i)\\
T_{x}(i) & =x\dot{=}i & T_{x}(N\varphi) & =\forall y(N(x,y)\rightarrow T_{y}(\varphi))\\
T_{x}(\varphi\wedge\psi) & =T_{x}(\varphi)\wedge T_{x}(\psi) & T_{x}(K\varphi) & =K_{x}T_{x}(\varphi)\\
T_{x}(\neg\varphi) & =\neg T_{x}(\varphi) & T_{x}(U\varphi) & =\forall xT_{x}(\varphi)
\end{align*}
The translation $T_{y}$ is obtained by exchanging $x$ and $y$ in
$T_{x}$.
\end{definition}

To show the translation truth-preserving, we embed the class of hybrid
network models into a class of term-modal models:
\begin{definition}
\label{def: mathsf.T.(M)} Let $M=(A,W,(N_{w})_{w\in W},\sim,g,V)$ be a
hybrid network model for $\mathcal{L}(P,Nom)$. Then the \textbf{$\mathsf{TML}$
image} of $M$ is the $\mathcal{L}(\Sigma_{n}(P,Nom))$ $\mathsf{TML}$
model $\mathsf{T}(M)=(A,W,\sim,I)$ sharing $A,W$ and $\sim$ with
$M$ and with $I$ given by
\begin{enumerate}
\item $\forall\underline{c}\in\mathtt{C},\forall w,v\in W,\forall a,b\in A,(I(\underline{c},w)=a\text{ and }w\sim_{b}v\Rightarrow I(\underline{c},v)=a)$
\item $I(p,w)=\{a\colon(w,a)\in V(p)\}$
\item $I(N,w)=\{(a,b)\in A\times A\colon(a,b)\in N_{w}\}$
\end{enumerate}
\end{definition}

\noindent The model $\mathsf{T}(M)$ has the same agents, worlds and
epistemic relations as $M$. The interpretation 1. encodes \textbf{\emph{weak
rigidity}}: if $(w,v)\in\bigcup_{a\in A}\sim_{a}$, then any constant
denotes the same in $w$ and $v$, emulating the rigid names of hybrid
network models; 2. ensures predicates are true of the same agents at
the same worlds, and 3. ensures the same agents are networked in the same worlds.

With the translations $T_{x},T_{y}$ and the embedding $\mathsf{T}$,
it may be shown that $\mathsf{DTML}$ can fully code the static semantics
of $\mathcal{L}(P,Nom)$ hybrid network logics:
\begin{proposition}
\label{prop: truth-preserving-static} Let $M=(A,W,(N_{w})_{w\in W},\sim,g,V)$
be a hybrid network model. Then for all $\varphi\in\mathcal{L}(P,Nom)$,
$M,w,g(\bullet)\models\varphi\text{ iff }\mathsf{T}(M),w\models_{g}T_{\bullet}(\varphi)$,
$\bullet=x,y$.
\end{proposition}

\subsection{\textsf{KDL} Dynamic Transformations and Learning Updates in \textsf{DTML}}

We show that $\mathsf{KDL}$ \cite{ChristoffHansenProietti2016} dynamics \textit{\emph{may be embedded
in}}\emph{ }$\mathsf{DTML}$, for finite agent sets (as assumed in
\cite{ChristoffHansenProietti2016}). Given Prop. \ref{prop: truth-preserving-static}%
, we argue that each $\mathsf{KDL}$ model transformer is representable
by a $\mathsf{DTML}$ action model and that the dynamic $\mathsf{KDL}$
language is truth-preservingly translatable into a $\mathsf{DTML}$
sublanguage. The logic of the class of $\mathsf{KDL}$ models is,
up to language translations, the logic of its corresponding class
of $\mathsf{DTML}$ models. We show that \textit{the logic of this
class of $\mathsf{DTML}$ models can be completely axiomatized, and
the resulting system is decidable}. Thus, by embedding $\mathsf{KDL}$
in $\mathsf{DTML}$, we find a complete system for the former.

In $\mathsf{KDL}$\footnote{Notation here is equivalent but different to fit better with the rest
of this paper.}, agents are described by \emph{feature propositions }reading ``for
feature $\mathsf{f}$, I have value $\mathsf{z}$''. With $\mathsf{F}$
a countable set of features and $\mathsf{Z_{f}}$ a finite set of
possible values of $\mathsf{f}\in\mathsf{F}$, the set of feature
propositions is $\mathsf{FP}=\{(\mathsf{f}\doteqdot\mathsf{z})\colon\mathsf{f}\in\mathsf{F},\mathsf{z}\in\mathsf{Z_{f}}\}$.\linebreak{}
The static language of \cite{ChristoffHansenProietti2016} is then
$\mathcal{L}_{-U}(\mathsf{FP},Nom)$. The dynamic language $\mathcal{L}_{\mathsf{KDL}}$
extends $\mathcal{L}_{-U}(\mathsf{FP},Nom)$ with dynamic modalities
$[d]$ and $[\ell]$ for \emph{dynamic transformations }$d$ and \emph{learning
updates} $\ell$:
\[
\varphi\Coloneqq(\mathsf{f}\doteqdot\mathsf{z})\mid i\mid\neg\varphi\mid\varphi\wedge\varphi\mid@_{i}\varphi\mid N\varphi\mid K\varphi\mid[d]\varphi\mid[\ell]\varphi
\]
A \textbf{dynamic transformation} $d$ changes feature values of agents:
each is a pair $d=(\Phi,post)$ where $\Phi\subseteq\mathcal{L}_{\mathsf{KDL}}$
is a non-empty finite set of pairwise inconsistent formulas and $post:\Phi\times\mathsf{F}\to(Z_{n}\cup\{\star\})$
is a $\mathsf{KDL}$ post-condition. Encoded by $post(\varphi,\mathsf{f})=x$
is the instruction: if $(w,a)\vDash\varphi$, then after $d$, set
$\mathsf{f}$ to value $x$ at $(w,a)$, if $x\in Z_{n}$; if $x=\star$,
$\mathsf{f}$ is unchanged. A \textbf{learning update} cuts accessibility
relations: the update with finite $\ell\subseteq\mathcal{L}_{\mathsf{KDL}}$
keeps a $\sim_{a}$ link between worlds $w$ and $v$ iff, for all
$\varphi\in\ell$, $(w,b)\vDash\varphi\Leftrightarrow(v,b)\vDash\varphi$
for all neighbors $b$ of $a$. 
\begin{definition}
Given a $\mathsf{KDL}$ model $M = (A,W, (N_w)_{w\in W}, \sim, g, V)$, the model reached after applying $d$ is $M^d = (A^d,W^d, (N^d_w)_{w\in W}, \sim^d, g^d, V^d)$ where only $V^d$ is different, and is defined as follows: $(w,a)\in V^d(\mathsf{f}\doteqdot\mathsf{z})$ iff (a)  $post(\varphi,\mathsf{f})=x$ for some $\varphi\in \Phi$ such that $M,w,a\models \varphi$, where $x \neq\star$; or (b) condition (a) does not hold and $(w,a)\in V(\mathsf{f}\doteqdot\mathsf{z})$. 
\end{definition}
\begin{definition} A learning update is a finite set of formulas $\ell \subseteq \mathcal{L}_{\mathsf{KDL}}$. Given a $\mathsf{KDL}$ model $M = (A,W, (N_w)_{w\in W}, (\sim_a)_{a\in A}, g, V)$, the model after $\ell$ is $M^\ell = (A,W, (N_w)_{w\in W}, (\sim'_a)_{a\in A}, g, V)$ where: 
\[w \sim'_a v \text{ iff }  w \sim_a v \text{ and } \forall b\in A (N_w(a,b) \Rightarrow \forall \varphi\in \ell ( M,w,b\models \varphi \text{ iff } M,v,b\models \varphi))\]
\end{definition}
Let $\mathsf{D}$ and $\mathsf{L}$ be the sets of dynamic transformations and learning updates. The result
of applying $\dagger\in\mathsf{D}\cup\mathsf{L}$ to $M$ is denoted
$M^{\dagger}$, and the $[\dagger]$ modality has semantics $M,w,a\models[\dagger]\varphi\text{ iff }M^{\dagger}w,a\models\varphi$. 

As we show below, for every $\dagger\in\mathsf{D}\cup\mathsf{L}$, there is a pointed $\mathsf{DTML}$ action model $\Delta^{\dagger}$ with
identical effects. As $\mathsf{KDL}$ operations may involve formulas
with $[\dagger]$-modalities, we must use $\mathsf{DTML}$ action
models that allow $[\Delta,e]$-modalities in their conditions, and
translate $\mathcal{L}_{\mathsf{KDL}}$ into the general $\mathsf{DTML}$
language that results, denoted $\mathcal{L}(\Sigma_{n}(\mathsf{FP},Nom)+[\Delta])$.\footnote{Defined using double recursion as standard; see \cite{ALR} for details.}
This language is interpreted over $\mathsf{DTML}$ models with standard action
model semantics:
\[
(M,w)\vDash_{g}[\Delta,e]\varphi\text{ iff }M\otimes\Delta,(w,e)\vDash\varphi.
\]

We define now the action models $\Delta^\dagger$. For a dynamic transformation $d\in\mathsf{D}$, \cite{ChristoffHansen2015} provide reduction axioms showing $d$'s instructions statically encodable
in $\mathcal{L}_{\mathsf{KDL}})$. The reduction axiom for atoms is as follows: 
\[[d]\mathsf{f}\doteqdot\mathsf{z} \leftrightarrow \left(\bigvee_{\varphi\in\Phi:post(\varphi,\mathsf{f}) = \mathsf{z},  \mathsf{z}\in \mathsf{Z_{f}}} \varphi\right) \vee \left(\neg \left( \bigvee_{\varphi\in\Phi: post(\varphi,\mathsf{f}) = \mathsf{z},  \mathsf{z}\in \mathsf{Z_{f}}} \varphi\right) \wedge \mathsf{f}\doteqdot\mathsf{z}\right)\]

As $d$ changes atomic truth values under a definable instruction, its effects may be simulated by an action model with a matching post-condition (i.e., the translation of the definable instruction). More specifically, the action model $\Delta^d$ is defined as follows.
\begin{definition} For dynamic transformation $d = (\Phi, post)$, the action model $\Delta^d= (E,Q,\mathsf{pre},\mathsf{post})$ is defined by $E = \{e^d\}$,  $Q(e^d,e^d) = \mathsf{pre}(e^d) = \top$ and for each constant $\underline{a}$, $\mathsf{post}(e)(T_x(\mathsf{f}\doteqdot\mathsf{z})(x \mapsto \underline{a})) =$ 

\[T_x\left(\left(\bigvee_{\varphi\in\Phi:post(\varphi,\mathsf{f}) = \mathsf{z},  \mathsf{z}\in \mathsf{Z_{f}}} \varphi\right) \vee \left(\neg \left( \bigvee_{\varphi\in\Phi: post(\varphi,\mathsf{f}) = \mathsf{z},  \mathsf{z}\in \mathsf{Z_{f}}} \varphi\right) \wedge \mathsf{f}\doteqdot\mathsf{z}\right)\right)(x \mapsto \underline{a})\]
\end{definition}

For a learning update $\ell\in\mathsf{L}$, $\Delta^{\ell}$
has events $e^{X},e^{Y}$ for any consistent subsets $X,Y$ of $\{\varphi(\underline{c}),\neg\varphi(\underline{c})\colon\varphi\in\ell,\underline{c}\in\mathtt{C}\}$ with edge-condition $Q(e^{X},e^{Y})$ satisfied for agents for whom all neighbors agree on $X$ and $Y$. Unsatisfied edge-conditions thereby capture the link cutting mechanism of $\ell.$ The detailed definition of $\Delta^\ell$ is as follows.

\begin{definition} Let $\ell=\{\varphi_1,\dots,\varphi_m\}$ be a learning update. Let $T_x(\ell) \coloneqq \{ T_x(\varphi_i) \mid i = 1,\dots,n\}$ and let $G_\ell \coloneqq \{ T_x(\varphi)(x\mapsto \underline{a}) \mid T_x(\varphi)\in T_x(\ell),\underline{a}\in\mathtt{C}\}$ be the grounding of $T_x(\ell)$ obtained by replacing each free occurrence of $x$ in $T_x(\varphi)$ for each possible constant $\underline{a}\in\mathtt{C}$. Define a $G_\ell$-valuation as a function $val: G_\ell \to \{0,1\}$ and let $\mathcal{V}_\ell$ be the set of all such valuations.\end{definition}

\begin{definition}  Let $\ell$ be a learning update. The corresponding $\mathsf{DTML}$ action model $\Delta^\ell= (E^\ell,Q^\ell,\mathsf{pre}^\ell,\mathsf{post}^\ell)$ is defined by letting 
\begin{itemize}
\item $E^\ell = \{e^{val} \mid val \in \mathcal{V}_\ell \}$,
\item  $\mathsf{pre}^\ell(e^{val}) = \bigwedge \{\varphi \mid val(\varphi)=1\}\cup \{\neg \varphi \mid val(\varphi)=0\}$
\item $Q^\ell(e^{val},e^{val}) = \top$
\item $Q^\ell(e^{val},e^{val'}) = \bigwedge_{\{\underline{a} \in \mathtt{C} \mid \exists \varphi \in\ell \text{ s.t. } val(T_x(\varphi)(x\mapsto\underline{a}))\neq val'(T_x(\varphi)(x\mapsto\underline{a}))  \}} \neg N(x^\star, \underline{a})$, for any two distinct events $e^{val}$, $e^{val'}$
\item $\mathsf{post}^\ell(e)=id$ for all $e\in E^\ell$
\end{itemize}

\end{definition}
Note that the signature $\Sigma_n(FP,ANom)$ is defined to have finitely many constants $\mathtt{C}=\{\underline{a}_{1},\dots,\underline{a}_{n}\}$, and hence both $E$, the preconditions and the edge-conditions in $\Delta^\ell$ are finite, as required. The action model $\Delta^\ell$ works as follows. Each event $e^{val}$ corresponds to one way the agents can be with respect to $G_\ell$, as indicated by $val$. The edge conditions control how links get cut. Two worlds $(w,e^{val})$ and $(v, e^{val'})$ in the updated model will keep a link for the agent named $\underline{a}$, if any disagreement between $val$ and $val'$ does not concern a neighbor of $\underline{a}$. Or, equivalently, if all neighbors of $\underline{a}$ are identical with respect to $G_\ell$. Precisely this condition is encoded in $Q(e^{val},e^{val'}) $.

To formally state that the dynamics of $\dagger\in \mathsf{D}\cup \mathsf{L}$ are simulated by $\Delta^\dagger$, the following clauses are added to translation $T_{\bullet}$, for
$\bullet=x,y$: 
\begin{align*}
T_{\bullet}([d]\varphi) & =[\Delta^{d},e^{d}]T_{\bullet}(\varphi), \\
T_{\bullet}([\ell]\varphi) & =\bigwedge_{e\in E^\ell}(\mathsf{pre}^\ell(e) \to [\Delta^{\ell},e]T_{\bullet}(\varphi))
\end{align*}
where $(\Delta^{\dagger},e^{\dagger})$ is an action model implementations
of \textsf{$\dagger\in\mathsf{D}\cup\mathsf{L}$}. Then $\mathsf{KDL}$
statics and dynamics can be shown performable in $\mathsf{DTML}$:
\begin{proposition}
\label{prop: truth-preserving-dynamic} For any finite agent hybrid
network model $M$ with nominal valuation $g$ and $\varphi\in\mathcal{L}_{\mathsf{KDL}}$:
$M,w,g(\bullet)\models\varphi\text{ iff }\mathsf{T}(M),w\models_{g}T_{\bullet}(\varphi)$,
for $\bullet=x,y$.
\end{proposition}

\begin{proof} By induction on $\varphi$. We include  the cases for the dynamic modalities. 

Let $\varphi = [d]\psi$, where $d = (\Phi, post)$. We need to show that \[M,w,g(x) \models  [d]\psi \text{ iff } \mathsf{T}(M),w\models_g [\Delta^d,e^d]T_x(\psi)\] (the case for $T_y$ is analogous). Note that $M,w,g(x) \models [d]\psi$ iff $M^d,w,g(x)\models \psi$ iff (by i.h.) $\mathsf{T}(M^d),w\models_g T_x(\psi)$. We will show that $\mathsf{T}(M^d)$ and $\mathsf{T}(M)\otimes \Delta^d$ satisfy the same formulas. To prove this, we will show that there is a bounded morphism linking these two models (it is straightforward to show that term-modal formulas are preserved when this is the case, as in the propositional modal setting). Define $b: \mathsf{T}(W^d) \to \mathsf{T}(W^{\Delta^d})$ by $w \mapsto (w,e^d)$. We show that $b$ is a bounded morphism.
\begin{enumerate}
\item $w$ and $(w,e^d)$ satisfy the same basic formulas:

$\mathsf{T}(M^d),w\models_g T_x( \mathsf{f}\doteqdot\mathsf{z})$ 

iff (i.h.) $M^d, w, g(x) \models \mathsf{f}\doteqdot\mathsf{z}$ 

iff $M,w,g(x)\models [d] \mathsf{f}\doteqdot\mathsf{z}$ 

iff (reduction axiom for $[d] \mathsf{f}\doteqdot\mathsf{z}$) 

$M,w,g(x)\models \left(\bigvee_{\varphi\in\Phi: post(\varphi,\mathsf{f}) = \mathsf{z},  \mathsf{z}\in \mathsf{Z_{f}}} \varphi\right) \vee$ 

$\ \ \ \ \ \ \ \ \ \ \ \ \ \ \ \ \ \ \left(\neg \left( \bigvee_{\varphi\in\Phi: post(\varphi,\mathsf{f}) = \mathsf{z},  \mathsf{z}\in \mathsf{Z_{f}}} \varphi\right) \wedge  \mathsf{f}\doteqdot\mathsf{z} \right)$

iff (i.h., where we let $g(x) = a$ for some $a\in A$ named $\underline{a}$) 

$\mathsf{T}(M),w\models_g T_x(\left(\bigvee_{\varphi\in\Phi: post(\varphi,\mathsf{f}) = \mathsf{z},  \mathsf{z}\in \mathsf{Z_{f}}} \varphi\right) \vee$

$\ \ \ \ \ \ \ \ \ \ \ \ \ \ \ \ \ \ \left(\neg \left( \bigvee_{\varphi\in\Phi: post(\varphi,\mathsf{f}) = \mathsf{z},  \mathsf{z}\in \mathsf{Z_{f}}} \varphi\right) \wedge  \mathsf{f}\doteqdot\mathsf{z} \right))(x \mapsto \underline{a})$ 

iff (by definition of $\Delta^d$) $\mathsf{T}(M),w\models_g \mathsf{post}(e)(T_x( \mathsf{f}\doteqdot\mathsf{z})(x \mapsto \underline{a}))$ 

iff $\mathsf{T}(M)\otimes \Delta^d,(w,e^d)\models_g T_x( \mathsf{f}\doteqdot\mathsf{z})(x \mapsto \underline{a})$

iff (since $g(x) = a$ and $a$ is named $\underline{a}$) $\mathsf{T}(M)\otimes \Delta^d,(w,e^d)\models_g  \mathsf{f}\doteqdot\mathsf{z}$.

\item if $(w,v) \in \mathsf{T}(\sim^d_a)$ then $((w,e^d), (v,e^d))\in \mathsf{T}(\sim^{\Delta^d}_a)$: 

$(w,v)\in \mathsf{T}(\sim^d_a)$ iff $(w,v)\in \sim^d_a$ iff $(w,v)\in \sim_a$ iff $(w,v)\in \mathsf{T}(\sim_a)$ iff $(w,v)\in \mathsf{T}(\sim^{\Delta^d}_a)$ (since $\Delta^d$ does not change the accessibility relations).

\item if  $((w,e^d), (v',e^d))\in \mathsf{T}(\sim^{\Delta^d}_a)$ then there is $v$ such that $(w,v)\in \mathsf{T}(\sim^d_a)$ and $b(v) = (v',e^d)$: 

Reasoning as in step 2, $((w,e^d), (v',e^d))\in \mathsf{T}(\sim^{\Delta^d}_a)$ iff $(w,v')\in \mathsf{T}(\sim^d_a)$, and $b(v') = (v', e^d)$.

\end{enumerate}

Hence, $b$ is a bounded morphism, and $\mathsf{T}(M^d)$ and $\mathsf{T}(M)\otimes \Delta^d$ satisfy the same formulas. Thus,  $M,w,g(x) \models [d]\psi$ iff $M^d,w,g(x)\models \psi$ iff (by i.h.) $\mathsf{T}(M^d),w\models_g T_x(\psi)$ iff (bounded morphism) $\mathsf{T}(M)\otimes \Delta^d,(w,e^d)\models_g T_x(\psi)$ iff $\mathsf{T}(M),w\models T_x([d]\psi)$.

Next, let $\varphi = [\ell]\psi$. We need to show that \[M,w,g(x) \models  [\ell]\psi \text{ iff } \mathsf{T}(M),w\models_g \bigwedge_{e\in E^\ell}(\mathsf{pre}^\ell(e) \to [\Delta^\ell,e]T_x(\psi))\] (the case for $T_y$ is analogous). Note that $M,w,g(x) \models [\ell]\psi$ iff $M^\ell,w,g(x)\models \psi$ iff (by i.h.) $\mathsf{T}(M^\ell),w\models_g T_x(\psi)$. As in the previous case, we will show that $\mathsf{T}(M^\ell)$ and $\mathsf{T}(M)\otimes \Delta^\ell$ satisfy the same formulas by defining a bounded morphism linking the two. Note that the preconditions in $\Delta^\ell$ are pairwise inconsistent and jointly exhaustive, since each precondition corresponds to one way of assigning truth values to the formulas in $G_\ell$. Hence, for each $w\in \mathsf{T}(W)$, there is exactly one event $e^{val}$ such that $\mathsf{T}(M), w\models \mathsf{pre}^\ell(e^{val})$. Define $b: \mathsf{T}(W^\ell) \to \mathsf{T}(W^{\Delta^\ell})$ by $w \mapsto (w,e^{val})$. We show that $b$ is a bounded morphism.

\begin{enumerate}
\item $w$ and $(w,e^{val})$ satisfy the same basic formulas: 

This is clear from the fact that learning updates do not change the accessibility relations. $\mathsf{T}(M^\ell),w\models_g T_x( \mathsf{f}\doteqdot\mathsf{z})$ iff (i.h.) $M^\ell, w, g(x) \models \mathsf{f}\doteqdot\mathsf{z}$ iff  $M,w,g(x)\models \mathsf{f}\doteqdot\mathsf{z}$ iff (i.h.) $\mathsf{T}(M),w\models_g T_x(\mathsf{f}\doteqdot\mathsf{z})$ iff $\mathsf{T}(M)\otimes \Delta^\ell,(w,e^{val})\models_g T_x(\mathsf{f}\doteqdot\mathsf{z})$.

\item if $(w,v) \in \mathsf{T}(\sim^\ell_a)$ then $((w,e^{val}), (v,e^{val'}))\in \mathsf{T}(\sim^{\Delta^d}_a)$: 

As $\mathsf{T}(M)$ is weakly rigid, each agent has the same name in each equivalence class $[w]_{\sim_a}$ of $\sim_a$. In what follows, we let the name of any agent $o\in A$ in worlds of $[w]_{\sim_a}$ be $\underline{o}$. Now, $(w,v)\in \mathsf{T}(\sim_a^{\ell})$ iff  $w \sim_a^{\ell} v$ 

iff  $w \sim_a v \text{ and } \forall b\in A (N_w ab \Rightarrow \forall \varphi\in \ell ( M,w,b\models \varphi \text{ iff } M,v,b\models \varphi))$ 

iff  (contrapositive) $w \sim_a v \text{ and } \forall b\in A( \exists \varphi\in \ell( (M,w,b\models \varphi \text{ and } M,v,b\models \neg \varphi) \text{ or } (M,w,b\models \neg\varphi \text{ and } M,v,b\models \varphi))  \Rightarrow \neg N_w ab )$

iff (by i.h.) $(w,v)\in \mathsf{T}(\sim_a)$ and (by def. of $\mathsf{T}(M)\otimes\Delta^\ell$) 

$\mathsf{T}(M),w\models_g \mathsf{pre}(e^{val})$ and $\mathsf{T}(M),v\models_g  \mathsf{pre}(e^{val'})$ for some $val,val' \in \mathcal{V}_\ell$, and for all $\underline{b}\in\mathsf{C}$:

if there is a $\varphi\in \ell$ such that
\begin{align*}
\MoveEqLeft[3]
\bigl(\mathsf{T}(M),w\models_g T_x(\varphi)(x\mapsto \underline{b}) \text{ and } \mathsf{T}(M),v\models_g T_x(\neg \varphi)(x\mapsto \underline{b})\bigr)  \\
  & \text{or } \bigl(\mathsf{T}(M),w\models_g T_x(\neg\varphi)((x\mapsto \underline{b}))\text{ and } {M,v\models_gT_x(\varphi)(x\mapsto \underline{b})}\bigr)
\end{align*}
then 
\[\mathsf{T}(M),w\models_g \neg N(\underline{a},\underline{b})\]

 iff $(w,v)\in \mathsf{T}(\sim_a)$ and $\mathsf{T}(M),w\models_g \mathsf{pre}(e^{val})$ and $\mathsf{T}(M),v\models_g  \mathsf{pre}(e^{val'})$ for some $val,val' \in \mathcal{V}_\ell$ and (by def. of $\Delta^\ell$) $\mathsf{T}(M),w\models_{g[x^\star \mapsto a]} Q(e^{val}, e^{val'})$ iff $((w,e^{val}),(v,e^{val'}))\in  \mathsf{T}(\sim^{\Delta^\ell}_a)$.

\item if  $((w,e^{val}), (v',e^{val'}))\in \mathsf{T}(\sim^{\Delta^\ell}_a)$ then there is $v$ such that $(w,v)\in \mathsf{T}(\sim^\ell_a)$ and $b(v) = (v',e^{val'})$: 

Reasoning as in step 2, $((w,e^{val}), (v',e^{val'}))\in \mathsf{T}(\sim^{\Delta^d}_a)$ iff $(w,v')\in \mathsf{T}(\sim^\ell_a)$, and $b(v') = (v', e^{val'})$. 
\end{enumerate}
Hence, $b$ is a bounded morphism, and $\mathsf{T}(M^\ell)$ and $\mathsf{T}(M)\otimes \Delta^\ell$ satisfy the same formulas. Thus, $M,w,g(x) \models [\ell]\psi$ iff $M^d,w,g(x)\models \psi$ iff (by i.h.) $\mathsf{T}(M^\ell),w\models_g T_x(\psi)$ iff (bounded morphism) for the unique event $e^{val}$ such that $\mathsf{T}(M),w \models_g \mathsf{pre}^\ell(e^{val})$, we have $\mathsf{T}(M)\otimes \Delta^\ell, (w, e^{val}) \models_g T_x(\psi)$ iff  $\mathsf{T}(M), w\models_g \bigwedge_{e\in E^\ell}(\mathsf{pre}^\ell(e) \to [\Delta^{\ell},e]T_{x}(\varphi)$ iff $\mathsf{T}(M)\models T_x([\ell]\psi)$.

This completes the proof.

\end{proof}

With Prop. 2 embedding $\mathsf{KDL}$
in $\mathsf{DTML}$, it remains to show that there is a complete and
decidable system for the image of $\mathsf{KDL}$. Up to translation,
such a logic is then a logic for the class of $\mathsf{KDL}$ models. To state the result, denote the $\mathsf{TML}$ image of the class
of $n$-agent $\mathsf{KDL}$ models by $\mathsf{T}(\mathsf{KDL}_{n})$. We now define a set of formulas, $\mathsf{F_n}$, which can be shown to characterise the class $\mathsf{T}(\mathsf{KDL}_{n})$.

\begin{definition}
Let $\mathsf{F}_{n}\subseteq\mathcal{L}(\Sigma_{n}(\mathsf{FP},Nom)+[\Delta])$
be the logic extending the term-modal $\mathsf{S5}$ logic with the reduction axioms for action models $(\Delta^{\dagger},e^{\dagger}),$
$\dagger\in\mathsf{D}\cup\mathsf{L}$ (defined in  \cite{ALR}), as well as the following static axioms:
\begin{enumerate} 
\item There are $n$ agents and they are all named:
\begin{align*}
\mathsf{Named_{n}}\coloneqq\exists x_{1},...,x_{n} & ( \left(\bigwedge_{i,j\leq n,i\neq j}x_{i}\neq x_{j}\right)\wedge\forall y\left(\bigvee_{i\leq n}y=x_{i}\right)\ \wedge\\
 & \left(\bigwedge_{i,j\leq n,i\neq j}c_{i}\neq c_{j}\right)\wedge\left(\bigwedge_{i\leq n}x_{i}=c_{i}\right) )
\end{align*}
\item Weak rigidity (Def. 8):
\[ \mathsf{Rig_{n}}\coloneqq  \bigwedge_{c\in\mathtt{C}}\forall x ((c=x) \rightarrow \forall y (K_y (c=x)))\]
\item The neighbour relation is irreflexive and symmetric: 
\[\mathsf{Neigh}\coloneqq \forall x \forall y (\neg N(x,x) \wedge (N(x,y)\leftrightarrow N(y,x)))\]
\item Agents know their neighbors: $\mathsf{KnowNeigh}\coloneqq \forall x \forall y(N(x,y) \leftrightarrow K_x N(x,y))$
\end{enumerate}
\end{definition}

We then obtain the result:

\begin{proposition} \label{prop: char} $\mathsf{F}_{n}$ statically
characterizes $\mathsf{T}(\mathsf{KDL}_{n})$.

\begin{proof}
By model-checking of the formulas in $\mathsf{F}_{n}$.
\end{proof}
\end{proposition}

Which we can use to state completeness:

\begin{theorem}
For any $n\in\mathbb{N}$, the logic $\mathsf{F}_{n}$ is sound, strongly
complete and decidable w.r.t. $\mathsf{T}(\mathsf{KDL}_{n})$.
\end{theorem}

\begin{proof}
[sketch] By Prop. \ref{prop: char}, $\mathsf{F}_{n}$ statically
characterizes $\mathsf{T}(\mathsf{KDL}_{n})$. The result then follows
from three results from \cite{ALR}: 1. Any extension of the term-modal
logic $\mathsf{K}$ with axioms $\mathsf{A}$ is strongly complete
with respect to the class of frames characterized by $\mathsf{A}$,
and 2. If $\mathsf{A}$ characterizes a class with finitely many agents,
then the logic is also decidable, and 3. any dynamic $\mathsf{DTML}$
formula is provably equivalent to a static $\mathsf{DTML}$ formula
using reduction axioms. 

Thus, since $\mathsf{F}_{n}$ characterizes $\mathsf{T}(\mathsf{KDL}_{n})$, which is a class with finitely many agents, and all dynamic axioms in $\mathsf{F}_n$ are probably equivalent to static $\mathsf{DTML}$, it follows that $\mathsf{K}+ \mathsf{F}_n$ is strongly complete and decidable with respect to $\mathsf{T}(\mathsf{KDL}_{n})$.
\end{proof}

\section{Final Remarks\label{sec:final}}
This paper has showcased $\mathsf{DTML}$ as a framework for modeling
social networks, their epistemics and dynamics, including examples
in which uncertainty about name reference and \emph{de dicto/de re }distinctions are key to modelling information flow and network change correctly.
It was shown that $\mathsf{DTML}$ may encode the popular hybrid logical
models of epistemic networks, and that $\mathsf{DTML}$ may be used
to obtain completeness for an open-question dynamics through emulation.

We are very interested in learning how $\mathsf{DTML}$ relates to
$\mathsf{GDDL}$ with respect to the encodable dynamics. We have been
able to emulate the updates used in the examples of \cite{Seligman2013}, but the general question is open. Further, the
statics of frameworks that describe networks using propositional logic
\cite{ChristoffRendsvig2014,Baltag_etal_2018} must be $\mathsf{DTML}$
encodable, and we expect the name about their updates, where reduction axioms exist. This raises two questions: if we can show this by a general
results instead of piecemeal, and whether principled $\mathsf{DTML}$
action models exist for classes of updates. E.g., the \emph{threshold
update }of \cite{Baltag_etal_2018} gives an agent's property $P$
if a given fraction of neighbors are $P$; for a fixed agent set,
this is $\mathsf{DTML}$ encodable by using the reduction axioms of
\cite{Baltag_etal_2018} to provide pre- and postconditions. For a
principled update, however, seemingly we need a generalized quantifier (e.g., a Rescher quantifier). If so, the general update form is not
$\mathsf{DTML}$ encodable. Classification results like these would add valuable insights on network logics.

\end{document}